\documentclass{article}
\usepackage{multirow}

\usepackage{amsmath}
\usepackage{amsthm}
\usepackage{amssymb}
\usepackage{cases}
\usepackage{bm}
\usepackage{breqn}
\usepackage[usenames,dvipsnames]{xcolor}

\newcounter{mnotecount}[section]

\newcommand{\mnotex}[1]
{\protect{\stepcounter{mnotecount}}$^{\mbox{\footnotesize $\bullet$\themnotecount}}$ 
\marginpar{
\raggedright\tiny\em
$\!\!\!\!\!\!\,\bullet$\themnotecount: #1} }

\newtheorem{lemma}{Lemma}
\newtheorem{theorem}{Theorem}

\usepackage{tikz}
\usetikzlibrary{shapes,arrows}

\usepackage[colorlinks=true,linkcolor=blue,citecolor=red,urlcolor=Violet,bookmarksdepth=subsection,final]{hyperref}

\allowdisplaybreaks

  \newcommand{\del}{\partial}

  \newcommand{\id}{\mathrm{id}}
  
  \newcommand{\Af}{\operatorname{Af}}
  \newcommand{\K}{\mathrm{K}}
  \newcommand{\CK}{\mathrm{CK}}

\DeclareFlexCompoundSymbol{\coloneq}{Rel}{:=}

\title{Conformal Killing Initial Data}
\author{%
	Alfonso Garc\'{\i}a-Parrado$^{\sharp}$%
		\thanks{E-mail: alfonso@utf.mff.cuni.cz} {} and
	Igor Khavkine$^\flat$%
		\thanks{E-mail: khavkine@math.cas.cz}
\\[2ex]
	{\small $^\sharp$Institute of Theoretical Physics, Faculty of Mathematics and Physics,}\\
	{\small Charles University in Prague, V~Hole\v{s}ovi\v{c}k\'ach~2, 180~00 Praha 8, Czech Republic}\\
	{\small $^\flat$Institute of Mathematics of the Czech Academy of Sciences,}\\
	{\small \v{Z}itn{\'a} 25, 115 67 Praha 1, Czech Republic}
}
\date{}

\begin{document}
\maketitle

\begin{abstract}
We find necessary and sufficient conditions ensuring that the vacuum
development of an initial data set of the Einstein's field equations
admits a conformal Killing vector. We refer to these conditions as {\em
conformal Killing initial data} (CKID) and they extend the well-known
{\em Killing initial data} (KID) that have been known for a long time.
The procedure used to find the CKID is a classical argument, which is
reviewed and presented in a form that may have an independent interest,
based on identifying a suitable {\em propagation} identity and
checking the well-posedness of the corresponding initial value problem.
As example applications, we review the derivation of the KID conditions,
as well as give a more thorough treatment of the homothetic Killing
initial data (HKID) conditions than was previously available in the
literature.
\end{abstract}

\section{Introduction} \label{sec:intro}

It is an interesting observation, first made in various forms
in~\cite{berezdivin,MONCRIEF-KID1,MONCRIEF-KID2,COLL77}, that there exists a set of linear partial
differential equations (PDEs) defined on the background of an initial
data surface for Einstein's vacuum equations such that the solutions of
these PDEs are in bijection with the Killing vectors of the vacuum
spacetime evolved from this surface. Fittingly, this system of PDEs is
known as the \emph{Killing initial data (KID)} equations, so named
in~\cite{beig-chrusciel}. Later, the KID equations have been generalized
to cover Einstein's equations coupled to rather general kinds of
matter~\cite{racz-kid1, racz-kid2}. After~\cite{beig-chrusciel}, KID
equations have received a fair amount of attention in the mathematical
relativity literature.

Killing vectors are solutions of the Killing equation, a geometric PDE
on a Lorentzian (more generally, pseudo-Riemannian) geometry. A natural
question arises: what other geometric PDEs have analogous initial data
systems? A solution of such an initial data system on an initial data
surface (for Einstein's vacuum or other related equations) would give
rise to a unique solution of the corresponding geometric PDE in the bulk
geometry evolved from the initial data surface. It seems that this
question has so far been considered in only a small number of cases. An
early and somewhat neglected example is~\cite{berger}, which treated
the initial data equations for homothetic Killing vectors and partially
for conformal Killing vectors. Rather recently, the valence $(1,0)$,
$(0,1)$, $(2,0)$ and $(0,2)$ Killing spinor equations in $4$~dimensions
were also treated~\cite{GOMEZLOBO2008}. In a follow-up
work~\cite{gasperin-williams}, some of them were adapted to Friedrich's
conformal vacuum equations. Also, Killing $(2,0)$-spinor initial data
equations have found applications to the characterization of initial
data for the Kerr black hole family~\cite{KERR-INVARIANT-TJ,
KERR-INVARIANT-PRL, KERR-INVARIANT-PRS, racz-etal}, and in the study of
the manifold structure of the infinite dimensional space of initial data
for Einstein's equations~\cite{bcs-kids}. 

Conceivably, such initial data equations (or at least the methods used
to obtain them) could also find applications in the study of the
uniqueness and rigidity of asymptotically flat black holes. For example,
such rigidity results were discussed in~\cite{IONESCU-KLAINERMAN-K} and,
while they did not directly use KID equations, they did use a
propagation equation (see below) analogous but not identical to the
standard one we later give in Section~\ref{sec:kid}.

In this work, we obtain for the first time a complete derivation of the
\emph{conformal Killing initial data (CKID)} equations, that is, whose
solutions on an initial data surface are in bijection with conformal
Killing vectors on the Einstein vacuum geometry (in any number of
dimensions, $n>2$) evolved from this surface. This question was first
approached in the early work~\cite{berger}, which gave some necessary
conditions on the initial data of a conformal Killing vector, but
stopped short of giving a complete list and, a fortiori, did not prove
the sufficiency of any such list. We make a more detailed comparison
with our results in the introductions to Sections~\ref{sec:hkid}
and~\ref{sec:ckid}. Some of the ideas from~\cite{berger} were picked up
again in~\cite{eimm-ckv}, and some classes of exact solutions to the
corresponding initial data equations were studied
in~\cite{zannias,sharma}. But no progress on the CKID problem appears to
have been been made since then. In~\cite{PAETZ201651}, the author
obtained the transformation of the KID system under a conformal
transformation of the bulk geometry under the assumption that that the
metric conformal equations of Friedrich are fulfilled in the bulk but
did not obtain what we call the CKID system.

We expect our CKID equations to have applications in mathematical
relativity analogous to the ones already mentioned for other initial
data systems. The CKID equations may be particularly useful when coupled
with Friedrich's conformal version of Einstein
equations~\cite{FRIEDRICH83}, or the equivalent system of conformally
covariant nonlinear wave equations~\cite{Paetz2015,vk-etal}. When
restricted to $4$~dimensions, the conformal Killing and Killing
$(1,1)$-spinor equations are equivalent. Hence the spinorial version of
the CKIDs could have been extracted from the intermediate results
of~\cite{GOMEZLOBO2008}, but only in $4$~spacetime dimensions. 

Our method of proof follows the same basic strategy as the old work on
the Killing equation~\cite{COLL77}. It relies on a key identity, which
we call a \emph{propagation equation}. This strategy is summarized in
Section~\ref{sec:propeq}, where the main observation is
Lemma~\ref{lem:propeq}, with the generic form of the desired key
identity expressed in Equation~\eqref{eq:propeq}. In
Section~\ref{sec:kid} we recall how the KID equations are derived, as
well as introduce some notation that is heavily used in subsequent
sections. In Section~\ref{sec:hkid}, we apply the same strategy to
\emph{homothetic Killing initial data} (HKID), confirming that the
initial data conditions first obtained in~\cite{berger} are in fact
sufficient. Finally, in Section~\ref{sec:ckid} we follow an analogous
route to obtain the CKID equations (Theorem~\ref{thm:ckid}).

It is also worth noting that the propagation equation identities giving
rise to KID, HKID and CKID systems are covariantly constructed and their form
does not explicitly depend on the signature of the metric tensor. Thus,
they would apply also in other signatures, like in Riemannian geometry.
In Lorentzian signature, we expect the propagation equations to be
hyperbolic and hence have a well-posed initial value problem. On the
other hand, in Riemannian signature, we expect the propagation equations
to be elliptic and hence have a well-posed boundary value problem. Then,
the bulk Killing or conformal Killing vectors will still induce
solutions of the KID or CKID equations on the boundary, but there may
exist solutions on the boundary that do not correspond to bulk solutions
when the elliptic propagation equations have non-trivial solutions for
homogeneous boundary conditions. The uniqueness of the (trivial)
solution for elliptic homogeneous boundary value problems may be
guaranteed using the Bochner method~\cite{bochner-book}, or some other
technique. Under such hypotheses, then the existence of Killing or
conformal Killing vectors on Ricci flat Riemannian manifolds with
boundary could be predicted by the existence of solutions of KID, HKID
or CKID equations with respect to the boundary data. It seems that such
applications have not yet been considered in Riemannian geometry.

All the computations of this paper have been double-checked with the
tensor computer algebra systems {\em Cadabra} and {\em
xAct}~\cite{CADABRA, cadabra2, XACT,XPERM}.

\section{Propagation equations and initial data} \label{sec:propeq}

From now on, all or our differential operators are presumed to be
defined between vector bundles over a manifold $M$ and have smooth
coefficients.

We call a linear partial differential equation (PDE) $P[\psi]=0$ a
\emph{propagation equation (of order $k\ge 1$)} if it has a well-posed
initial value problem: given a Cauchy surface $\Sigma \subset M$ with
unit normal $n^a$, the equation can be put into Cauchy-Kovalevskaya form
(solved for the highest time derivative) and for each assignment of arbitrary
smooth initial data ${\psi|_{\Sigma}}=\psi_0, \ldots , \nabla_n^{k-1}
{\psi|_{\Sigma}}=\psi_{k-1}$ (where $\nabla_n = n^a\nabla_a$) there
exists a unique solution of $P[\psi]=0$ on all of $M$. In
particular, due to the linearity of the propagation equation, if the
initial data all vanish, $\psi_0 = \cdots = \psi_{k-1} = 0$, then
$\psi=0$ is the corresponding unique solution on $M$.

There are multiple examples of propagation equations: (a) Wave (a.k.a\ 
\emph{normally-hyperbolic}) equations, $P[\psi] = \square \psi +
P'(\nabla\psi,\psi)$~\cite{bgp}. (b) Transport equations, $P[\psi] = u^a
\nabla_a \psi + P'(\psi)$, with $u^a$ everywhere transverse to
$\Sigma$~\cite{IONESCU-KLAINERMAN-K}. (c) Special cases, like
$P_{bcd}[\psi] = \nabla^a \psi_{abcd}$ for $\psi_{abcd}$ satisfying the
symmetry and tracelessness conditions of the Weyl tensor in
4~dimensions~\cite{IONESCU-KLAINERMAN-K}. (d) At the end of this section
(Lemma~\ref{lem:gen-normhyp}), we give the name \emph{generalized
normally-hyperbolic} to a class generalizing that in (a).

\begin{lemma} \label{lem:propeq}
Consider a globally hyperbolic spacetime $(M,g)$, satisfying the
Einstein vacuum equations, $G_{ab} = R_{ab} - \frac{1}{2} R g_{ab} = 0$.
Let $E[\phi] = 0$ be a PDE (system) defined on some (possibly
multicomponent) field $\phi$. Suppose that there exist propagation
equations $P[\psi] = 0$, $Q[\phi] = 0$ (of respective orders $k$ and
$l$), where the differential operators $P$ and $Q$ satisfy the identity
\begin{dmath} \label{eq:propeq}
	P[E[\phi]] = \sigma[Q[\phi]] + \tau[G] ,
\end{dmath}
for some linear differential operators $\sigma$ and $\tau$. Then, given
a Cauchy surface $\Sigma \subset M$ with unit timelike normal $n^a$, the
unique solution of $Q[\phi] = 0$ with initial data ${\phi|_\Sigma} =
\phi_0, \ldots, \nabla_n^{l-1} {\phi|_\Sigma} = \phi_{l-1}$ satisfies
the equation $E[\phi] = 0$ provided the initial data ${\psi|_{\Sigma}}=0,
\ldots , \nabla_n^{k-1} {\psi|_{\Sigma}}=0$, for $\psi = E[\phi]$,
vanish.

In addition, there exists a purely spatial linear PDE on $\Sigma$,
$E^\Sigma[\phi_0,\ldots,\phi_{l-1}] = 0$ such that the conditions
$Q[\phi] = 0$ and $E^\Sigma[{\phi|_\Sigma}, \ldots,
\nabla_n^{l-1}{\phi|_\Sigma}] = 0$ imply the vanishing of the initial data
${\psi|_{\Sigma}}=0, \ldots , \nabla_n^{k-1} {\psi|_{\Sigma}}=0$ for $\psi =
E[\phi]$.
\end{lemma}

\begin{proof}
Under the hypotheses on the metric $g$ and $\phi$, both $G=0$ and
$Q[\phi] = 0$ vanish. Then, letting $\psi = E[\phi]$, the
identity~\eqref{eq:propeq} implies $P[\psi] = 0$. But, by the definition
of a propagation equation, the vanishing of the initial data for $\psi$
on $\Sigma$ implies that $E[\phi] = \psi = 0$ on all of $M$.

For the second part, first notice that our notion of well-posedness for
the propagation equation $Q[\phi] = 0$ implies that the values of the
time derivatives $\nabla_n^N \phi$ for $N\ge l$ are given by local
algebraic expressions in terms of the $\nabla_n^N \phi$ for $0\le N <
l$. Setting $\psi = E[\phi]$, the vanishing of the initial data for
$\psi$ may a priori involve time derivatives $\nabla_n^N \phi$ of orders
$N\ge l$. But replacing these higher order time derivatives by the above
expressions, reduces the dependence on time derivatives $\nabla_n^N\phi$
of order at most $N<l$. Then, obviously, these reduced order conditions
can be collected into a single equation, which we can denote by
$E^\Sigma[\phi_0,\ldots,\phi_{l-1}] = 0$.
\end{proof}

Intuitively, the original equation $E[\phi] = 0$ should be more
restrictive than the propagation equation $Q[\phi] = 0$, thus requiring
the $E^\Sigma[{\phi|_\Sigma},\ldots]=0$ initial data constraints to make
up the difference. However, as written above, Lemma~\ref{lem:propeq}
does not exclude situations where the equation $Q[\phi] = 0$
is the more
restrictive one (like the extreme example $Q[\phi]=\phi$). Thus, when we
would like every solution of $E[\phi]=0$ to also be a solution of
$Q[\phi]=0$, we will refer to the following obvious

\begin{lemma} \label{lem:propeq-conv}
Using the notation of Lemma~\ref{lem:propeq}, suppose there exists a
linear differential operator $\rho$ such that
\begin{equation} \label{eq:propeq-conv}
	Q[\phi] = \rho[E[\phi]] .
\end{equation}
Then any solution of $E[\phi] = 0$ is also a solution of $Q[\phi] = 0$
with vanishing initial data ${\psi|_{\Sigma}}=0, \ldots , \nabla_n^{k-1}
{\psi|_{\Sigma}}=0$, for $\psi = E[\phi]$.
\end{lemma}

For an operator $E^\Sigma$ satisfying the second part of
Lemma~\ref{lem:propeq}, we call
\begin{equation}\label{eq:p-initial-data}
	E^\Sigma[\phi_0, \ldots, \phi_{l-1}]=0
\end{equation}
a set of \emph{$E$-initial data conditions} or a \emph{$E$-initial data
system}. Clearly, the operator $E^\Sigma$ is not uniquely fixed. For
instance, its components may contain many redundant equations. Thus, in
practice, once some $P$-initial data conditions have been obtained, they
will be significantly simplified by eliminating as many higher order (in
spatial derivatives) terms as possible. Also, when some of the
components of $E^\Sigma[\phi_0,\ldots,\phi_{l-1}]=0$ can be used to
directly solve for one of the arguments, say $\phi_{l-1}$, in terms of
the remaining ones, we can split the initial data system into
(a) $\phi_{l-1}=\cdots$ and (b) a system involving only the remaining arguments,
$E^{\prime\Sigma}[\phi_0,\ldots,\phi_{l-2}]=0$. When presenting an
initial data system, we will omit from $E^\Sigma$ those components that
can be rewritten as type (a) and only write the remaining components of
type (b), reduced to the smallest convenient set of arguments. Of
course, the derivation of the initial data system will provide the
information about how all type (a) components can be recovered.

There is a limited set of known examples of propagation identities for
geometrically motivated equations $E[\phi] = 0$ in Lorentzian (or
Riemannian) geometry. The most prominent example concerns the Killing
equation in any spacetime dimension (examined in detail in
Section~\ref{sec:kid})~\cite{beig-chrusciel}. The list of known examples
is then exhausted by the 4-spacetime dimensional Killing spinor
equations of valences $(1,0)$, $(0,1)$, $(2,0)$ and
$(0,2)$~\cite{GOMEZLOBO2008, gasperin-williams}. The propagation
identity ostensibly obtained for the homothetic Killing vector equation
in~\cite{berger} was not not formally checked for well-posedness. We
close this small gap in Section~\ref{sec:hkid}, where we check the
well-posedness of our propagation identity for this equation.

\medskip

Normally-hyperbolic equations~\cite{bgp}, mentioned earlier in this
section, are a large and easy to recognize class of propagation
equations. One need only check that the principal symbol of $Q[\phi] =
0$ coincides with that of the wave operator $\square$, possibly tensored
with the identity endomorphism of the vector bundle where the field
$\phi$ takes its values. However, there exist second order operators $Q$ whose
highest order terms consist of more than just $\square$, yet are closely
tied to the normally-hyperbolic class.

We will call \emph{generalized normally-hyperbolic} any operator $Q$ (of
order $l$) that is determined (acts between vector bundles of equal
rank) and for which there exists an operator $Q'$ (or order $2m-l$,
$m\ge 1$) such that
\begin{equation}
	N[\phi] := Q'[Q[\phi]] = \square^m \phi + \text{l.o.t} ,
\end{equation}
where l.o.t stands for term of differential order lower than $2m$.
That is, the principal symbol of $N[\phi]$ is a power of the wave
operator and hence $N$ is normally-hyperbolic.%
	\footnote{While~\cite{bgp} only treats \emph{second order}
	normally-hyperbolic equations, any such \emph{higher order} equation
	can be order reduced to a second order normally-hyperbolic
	\emph{system} of equations, which \emph{are} treated in~\cite{bgp}.} %
Generalized normally-hyperbolic operators will appear as propagation
operators in the study of the conformal Killing equation. Thus, for
later convenience, imitating the treatment of the Dirac operator
in~\cite{bgp, baer-greenhyp}, we establish the following
\begin{lemma} \label{lem:gen-normhyp}
Any generalized normally-hyperbolic operator has a well-posed initial
value problem.
\end{lemma}

\begin{proof}
Consider $(M,g)$ to be a globally hyperbolic Lorentzian manifold, with a
Cauchy surface $\Sigma \subset M$ with unit timelike normal $n^a$.
Suppose $Q$ is of order $l$ and generalized normally-hyperbolic, with
$Q'$ of order $2m-l$ such that $N := Q'\circ Q = \square^m +
\text{l.o.t}$, as in the definition. As we have noted already, $N$ is
normally-hyperbolic and hence has a well-posed initial value problem of
order $2m$.

There is an immediate consequence of the existence of such an operator
$Q'$. Namely, because of the condition on the differential orders of all
the operators, we know that $\sigma_p(N) = \sigma_p(Q') \sigma_p(Q)$,
where $\sigma_p(-)$ denotes the principal symbol of an operator (a
vector bundle morphism valued function on the cotangent bundle $T^*M \ni
p$). Now, because the principal symbol of $N$ coincides with that of
$\square^m$, we know that $\sigma_p(N)$ is invertible everywhere except
at null covectors $p \in T^*M$. Therefore, $\sigma_p(Q)$ is invertible
wherever $\sigma_p(N)$ is, in particular whenever $p$ is non-null, because
\begin{equation} \label{eq:symb-inv}
	\left[\sigma_p(N)^{-1} \sigma_p(Q')\right] \sigma_p(Q) = \id
	\quad \text{and} \quad
	\sigma_p(Q) \left[\sigma_p(N)^{-1} \sigma_p(Q')\right] = \id ,
\end{equation}
where the second equality holds because $Q$ is determined (its principal
symbol is a square matrix in components). Thus, the operator $Q$
may be expanded as
\begin{equation}
	Q[\phi] = \sigma_n(Q) (\nabla_n^l\phi) + \text{l.o.t}_n ,
\end{equation}
where l.o.t$_n$ stands for terms of lower differential order in normal
derivatives $\nabla_n$, and the notation $\sigma_n(-)$ stands for the
principal symbol evaluated specifically at the covector $n_a$, which is
non-null, being orthogonal to $\Sigma$. 

This means that, $Q[\phi]=0$ may be put into Cauchy-Kovalevskaya form.
In other words, whenever $Q[\phi] = 0$, we can write the normal
derivative $\nabla_n^l {\phi|_\Sigma}$ as a \emph{linear local} (meaning
as a purely spatial differential operator on $\Sigma$) expression of the
lower order normal derivatives, $\phi_0 = {\phi|_\Sigma}, \ldots,
\phi_{l-1} = \nabla_n^{l-1}{\phi|_\Sigma}$. Which means that, after
applying $\nabla_n$ multiple times to that relation, all normal
derivatives also up to $\nabla_n^{2m-1}\phi$ can be linearly locally
expressed in terms of $\phi_0, \ldots, \phi_{l-1}$ as well. In other
words, the initial data for $Q[\phi]=0$ uniquely determine the initial
data for $N[\phi]=0$, while the latter equation produces a unique
solution $\phi$ with those initial data. It remains to check that
$Q[\phi]=0$ is actually satisfied by this $\phi$. But to that end, we
need only apply Lemma~\ref{eq:propeq} to the obvious identity
\begin{equation}
	Q[N[\phi]] = Q[Q'[Q[\phi]]] = N'[Q[\phi]] ,
\end{equation}
where $N' := Q\circ Q'$ and we know that $N' = \square^m +
\text{l.o.t}$, as a consequence of the second equality
in~\eqref{eq:symb-inv}. As a technicality, we need to check that the
solution of $N[\phi]=0$ with the initial data constructed earlier gives
$\psi = Q[\phi]$ with vanishing initial data for $N'[\psi]=0$, namely
${\nabla_n^{2m-1}\psi|_\Sigma}=0, \ldots, {\psi|_\Sigma}=0$. For that to
be true, we just need to show that the two ways of solving for the
higher order derivatives $\nabla_n^{k+2m}\phi$ actually agree, namely
from solving $\nabla_n^k N[\phi]=0$ or $\nabla_n^{k+2m-l} Q[\phi]=0$ on
$\Sigma$. But applying to $Q'$ the argument from the first part of the
proof, we know that $Q'[\psi]=0$ can also be put in Cauchy-Kovalevskaya
form, so that by the identity
\begin{equation}
	\nabla_n^k N[\phi]
	= \nabla_n^k Q'[Q[\phi]]
	= \sigma_n(Q') (\nabla_n^{k+2m-l} Q[\phi] + \text{l.o.t}_n) ,
\end{equation}
the two ways are indeed equivalent.

Hence, since arbitrary initial data of order $l$ determine a unique
solution to $Q[\phi]=0$, this equation has a well-posed initial value
problem of order $l$.
\end{proof}

\subsection{Example: Killing initial data} \label{sec:kid}

The canonical illustration of Lemma~\ref{lem:propeq} is the case of the
\emph{Killing equation}~\cite{beig-chrusciel},
\begin{dgroup*}
\begin{dmath} \label{eq:k}
	\K_{ab}[v] \hiderel{=} \nabla_a v_b + \nabla_b v_a = 0
	\condition[]{$(E[\phi] = 0)$} .
\end{dmath}
\intertext{The corresponding propagation equations are}
\begin{dmath} \label{eq:k-v-propeq}
	\square v_a + R_a{}^b v_b = 0
	\condition[]{$(Q[\phi] = 0)$} ,
\end{dmath}
\begin{dmath} \label{eq:k-h-propeq}
	\square h_{ab} - 2R^c{}_{ab}{}^d h_{cd} = 0
	\condition[]{$(P[\psi] = 0)$} ,
\end{dmath}
\end{dgroup*}
where $h_{ab}$ is considered to be symmetric, while the propagation
identities~\eqref{eq:propeq} and~\eqref{eq:propeq-conv} take the form 
\begin{dgroup*}
\begin{dmath} \label{eq:k-propeq}
	\square \K_{ab}[v] - 2R^c{}_{ab}{}^d \K_{cd}[v]
		= \K_{ab}[\square v + R\cdot v]
		+ 2 R_{(a}{}^c \K_{b)c}[v]- 2 \mathcal{L}_{v} R_{ab} , \\
	{}\hspace{7em} (P[E[\phi]] \hiderel{=} \sigma[Q[\phi]] + \tau[G])
\end{dmath}
\begin{dmath} \label{eq:k-propeq-conv}
	\square v_a + R_a{}^b v_b
		= \nabla^b \K_{ab}[v] - \frac{1}{2} \nabla_a \K^b{}_b[v] , \\
	{}\hspace{7em} (Q[\phi] \hiderel{=} \rho[E[\phi]])
\end{dmath}
\end{dgroup*}
where we denoted $(R\cdot v)_a = R_a{}^b v_b$ and ${\mathcal L}_v R_{ab}
= v^c \nabla_c R_{a b} + 2R_{c (a}\nabla_{b)}v^c$ is the Lie derivative
of $R_{ab}$ with respect to the vector field $v$.

To obtain the $\K$-initial data conditions, or more commonly the
\emph{Killing initial data (KID)} conditions, we must first introduce a
space-time split around a Cauchy surface $\Sigma \subset M$, $\dim M =
n$ and $\dim \Sigma = n-1$. Let us use Gaussian normal coordinates to
set up a codimension-$1$ foliation on an open neighborhood $U \supset
\Sigma$ by level sets of a smooth temporal function $t\colon U\to
\mathbb{R}$, of which $\Sigma = \{ t = 0 \}$ is the zero level set.
Choose $t$ such that $n_a = \nabla_a t$ is a unit normal to the level
sets of $t$. Let us identify tensors on $\Sigma$ by upper case Latin
indices $A, B, C, \ldots$, denote the pullback of the ambient metric to
$\Sigma$ by $g_{AB}$ and its inverse by $g^{AB}$, and also denote by
$h^a_A$ the injection $T_\Sigma\to TM$ induced by the foliation. Raising
and lowering the respective indices on $h^a_A$ with $g_{ab}$ and
$g_{AB}$, we get the corresponding injections and orthogonal projections
between $T\Sigma$, $T^*\Sigma$, $TM$ and $T^*M$. In our notation, all
covariant and contravariant tensors split according to
\begin{equation}
	v_a = v_0 n_a + h_a^A v_A, \quad
	u^b = -u^0 n^b + h^b_B u^B ,
\end{equation}
which we also denote by
\begin{equation}
	v_a \to \begin{bmatrix} v_0 \\ v_A \end{bmatrix} , \quad
	u^b \to \begin{bmatrix} u^0 \\ u^B \end{bmatrix} .
\end{equation}
Thus, in our convention, the ambient metric splits as
\begin{equation}
	g_{a b} \to \begin{bmatrix} -1 & 0 \\ 0 & g_{A B} \end{bmatrix} .
\end{equation}
Let $D_A$ denote the Levi-Civita connection on $(\Sigma,g_{AB})$,
depending on the foliation time $t$ of course, and let $\del_t =
\mathcal{L}_{-n}$ denote the Lie derivative with respect to the
future-pointing normal vector $-n^a$. The action of $\del_t$ extends to
$t$-dependent tensors on $\Sigma$ in the natural way. The
($t$-dependent) extrinsic curvature on $\Sigma$ is then defined by
\begin{equation}
	\pi_{A B} = \frac{1}{2} \del_t g_{A B}
\end{equation}
and the ambient spacetime connection decomposes as
\begin{equation}
	\nabla_a v_b
	\to \begin{bmatrix} \nabla_0 v_b \\ \nabla_A v_b \end{bmatrix} ,
\end{equation}
where 
\begin{align}
	\nabla_0 v_a \to
		\begin{bmatrix}
			\nabla_0 v_0 \\ \nabla_0 v_A
		\end{bmatrix}
	&= \begin{bmatrix}
			\del_t & 0 \\
			0 & \del_t \delta_A^B - \pi_{A}{}^{B}
		\end{bmatrix}
		\begin{bmatrix}
			v_0 \\ v_B
		\end{bmatrix} ,
	\\
	\nabla_A v_b \to
		\begin{bmatrix}
			\nabla_A v_0 \\ \nabla_A v_B
		\end{bmatrix}
	&= \begin{bmatrix}
			D_A & -\pi_A{}^C \\
			-\pi_{A B} & D_A \delta_B^C
		\end{bmatrix}
		\begin{bmatrix}
			v_0 \\ v_C
		\end{bmatrix} .
\end{align}
The ambient vacuum Einstein equations $R_{ab} = 0$ decompose as
\begin{equation} \label{eq:ricci-split}
	R_{ab} \to
	\begin{bmatrix}
		-\nabla_0 \pi_{C}{}^{C} - \pi \cdot \pi
			& D^{C} \pi_{C B} - D_{B} \pi_{C}{}^{C} \\
		D^{C} \pi_{C A} - D_{A} \pi
			& \nabla_0 \pi_{AB} + \pi \pi_{AB} + r_{AB}
	\end{bmatrix} = 0 ,
\end{equation}
where now $r_{AB}$ is the Ricci tensor of $g_{AB}$ on $\Sigma$, $\pi =
\pi_C{}^C$, $(\pi\cdot \pi)_{AB} = \pi_A{}^C \pi_{C B}$ and $\pi\cdot
\pi = (\pi\cdot \pi)_C{}^C$. Note that we have found it convenient to
use the $\nabla_0$ operator instead of $\del_t$, because of its
preservation of both the orthogonal splitting with respect to the
foliation and of the spatial metric, $\nabla_0 g_{AB} = \nabla_0 g^{AB}
= 0$. For convenience, we note the commutator 
\begin{equation}\label{eq:commutator}
	(\nabla_0 D_A - D_A \nabla_0) \begin{bmatrix} v_0 \\ v_B \end{bmatrix}
	= -\pi_A{}^C D_C \begin{bmatrix} v_0 \\ v_B \end{bmatrix}
	+ \begin{bmatrix}
			0 \\ (D^C\pi_{AB} - D_B\pi_A{}^C)
		\end{bmatrix} v_C \;.
\end{equation}
According to Lemma~\ref{lem:propeq} and the specific
identity~\eqref{eq:k-propeq}, the Killing equation $\K_{ab}[v] = 0$ is
satisfied when $v_a$ is any solution of~\eqref{eq:k-v-propeq} where both
\begin{dgroup}
\begin{dmath}
	\left.\K_{ab}[v]\right|_\Sigma
	\to \left.\begin{bmatrix}
			\K_{00}[v] & \K_{0B}[v] \\
			\K_{0A}[v] & \K_{AB}[v]
		\end{bmatrix}\right|_\Sigma \hiderel{=} 0
\end{dmath}
\begin{dmath}
	\text{and} \quad
	\left.\nabla_0\K_{ab}[v]\right|_\Sigma
	\to \left.\begin{bmatrix}
			\nabla_0\K_{00}[v] & \nabla_0\K_{0B}[v] \\
			\nabla_0\K_{0A}[v] & \nabla_0\K_{AB}[v]
		\end{bmatrix}\right|_\Sigma \hiderel{=} 0 .
\end{dmath}
\end{dgroup}
In more detail, these components are
\begin{dgroup}
\begin{dmath}
	\K_{00}[v] = 2 \nabla_0 v_0 ,
\end{dmath}
\begin{dmath}
	\K_{0B}[v] = \nabla_0 v_B - \pi_{BC} v^C + D_B v_0 ,
\end{dmath}
\begin{dmath}
	\K_{AB}[v] = D_A v_B + D_B v_A - 2\pi_{AB} v_0 ,
\end{dmath}
\begin{dmath}
	\nabla_0 \K_{00}[v] = 2 \nabla_0\nabla_0 v_0 ,
\end{dmath}
\begin{dmath}
	\nabla_0 \K_{0B}[v] = \nabla_0\nabla_0 v_B - (\nabla_0 \pi_{BC}) v^C \\
		- \pi_{BC} \nabla_0 v_0
		+ D_B \nabla_0 v_0 - \pi_{BC} D^C v_0 ,
\end{dmath}
\begin{dmath}
	\nabla_0 \K_{AB}[v] = D_A \nabla_0 v_B + D_B \nabla_0 v_A
		+ 2(D^C \pi_{AB}) v_C
		- 2\pi_{C(A} D^C v_{B)}
		- 2D_{(A} \pi_{B)C} v^C
		- 2\pi_{AB} \nabla_0 v_0
		- 2(\nabla_0 \pi_{AB}) v_0 .
\end{dmath}
\end{dgroup}
On the other hand, the propagation equation~\eqref{eq:k-v-propeq}
splits (modulo $R_{ab} = 0$) as
\begin{dgroup}[indentstep=10em]
\begin{dmath} \label{eq:k-v-propeq-time}
	-\nabla_0\nabla_0 v_0
	+ D^{C}{D_{C}{v_{0}}}
	- 2 \pi^{BC} D_{B} v_{C}
	- \pi \nabla_0 v_{0} \\
	- (D_{B} \pi^{BC}) v_{C}
	+ (\pi\cdot\pi) v_{0}
	= 0 ,
\end{dmath}
\begin{dmath} \label{eq:k-v-propeq-space}
	-\nabla_0\nabla_0 v_{A}
	+ D^{C}{D_{C}{v_{A}}}
	- 2\pi_{A}{}^{B} D_{B} v_{0}
	- \pi \nabla_0 v_{A} \\
	- (D^{B}\pi_{A B}) v_{0}
	+ (\pi\cdot\pi)_{A C} v^{C}
	= 0 .
\end{dmath}

\end{dgroup}

Finally, eliminating the time derivatives of $v_0$ and $v_A$, while also
eliminating the time derivatives of $\pi_{AB}$ using the vacuum Einstein
equations~\eqref{eq:ricci-split}, we obtain the well-known Killing
initial data (KID) conditions:
\begin{dgroup}[indentstep=6em] \label{eq:kid}
\begin{dmath} \label{eq:kid0}
	D_A v_B + D_B v_A - 2\pi_{AB} v_0 = 0 ,
\end{dmath}
\begin{dmath} \label{eq:kid1}
	D_{A} D_{B} v_{0}
	+ (2(\pi\cdot\pi)_{A B} - \pi \pi_{AB} - r_{AB}) v_0 \\
	- 2 \pi_{(B}{}^{C} D_{A)} v_C
	- (D^{C}{\pi_{A B}}) v_{C} = 0 .
\end{dmath}
\end{dgroup}
For ease of comparison with the conformal and homothetic case in
Sections~\ref{sec:hkid} and~\ref{sec:ckid}, let us note the traces of
the above KID conditions:
\begin{dgroup} \label{eq:trkid}
\begin{dmath} \label{eq:trkid0}
	2 (D_C v^C - \pi v_0) = 0 ,
\end{dmath}
\begin{dmath}[indentstep=6em] \label{eq:trkid1}
	D^{C} D_{C} v_{0}
	- (\pi\cdot\pi) v_0
	- (D^{C} \pi) v_{C} \\
	- \pi^{EF} (D_E v_F + D_F v_E - 2\pi_{EF} v_0)
	= 0 ,
\end{dmath}
\end{dgroup}
where we have used the Hamiltonian constraint $r + \pi^2 - \pi\cdot\pi =
0$ from the vacuum Einstein equations to eliminate the spatial scalar
curvature $r$.

The propagation identity and Killing initial data equations recalled in
this section, coupled with Lemmas~\ref{eq:propeq}
and~\ref{lem:propeq-conv}, allow us in particular to identify those
initial data sets for the metric $g$ that give rise to a solution of the
Einstein vacuum equations equations with Killing symmetries. This was
the original motivation under which the Killing initial data
conditions~\eqref{eq:kid} were first identified~\cite{COLL77,
MONCRIEF-KID1, MONCRIEF-KID2, beig-chrusciel}. Under such
considerations, the propagation identity~\eqref{eq:k-propeq} could have
been simplified, by dropping any terms that vanish in vacuum. However,
similar condition have also been found to identify initial data sets
some non-vacuum solutions of Einstein equations with Killing
symmetries~\cite{racz-kid1, racz-kid2}, where the full propagation
identity~\eqref{eq:k-propeq} plays a crucial role.

\section{Homothetic Killing initial data} \label{sec:hkid}

Recall that \emph{homothetic Killing} vectors $v_a$ are those that
satisfy the \emph{conformal Killing} equation in addition to having a
constant divergence (together giving $E[\phi]=0$ in the notation of
Section~\ref{sec:propeq}), namely
\begin{dgroup*}
\begin{dmath} \label{eq:ck}
	\CK_{ab}[v] \coloneq \nabla_a v_b + \nabla_b v_a - \frac{2}{n} g_{ab} \nabla^c v_c
\end{dmath}
\begin{dmath} \label{eq:const-div}
	\nabla_a(\nabla^b v_b) = 0 .
\end{dmath}
\end{dgroup*}
The $\CK_{ab}[v]$ operator is simply the trace free part of the Killing
operator $\K_{ab}[v]$ from~\eqref{eq:k},
\begin{equation}
	\CK_{ab}[v] = \K_{ab}[v] - \frac{1}{n} g_{ab} \K_{c}{}^{c}[v] .
\end{equation}

The homothetic Killing vector equation is less restrictive than the
Killing equation itself, but is more restrictive than just the conformal
Killing vector equation. Homothetic Killing vectors are important in
relativity because of the interesting observation~\cite[Thm.3]{eimm-ckv}
that a conformal Killing vector on a 4-dimensional Einstein vacuum
spacetime (where $R_{ab}=0$) is almost always homothetic. This is a
purely local result, independent of the global features of a spacetime.
The only exceptions that admit proper conformal Killing vectors
(non-homothetic ones) are locally (a) flat spacetime or (b) a subclass of
type~N spacetimes. It is likely that analogous restrictions exist in
other dimensions and non-Lorentzian signatures.

Homothetic Killing vectors also play an essential role in the definition
of {\em (asymptotically) self-similar solutions} and therefore their
characterization from an initial data set point of view could be
relevant in the study of these solutions with applications to {\em
critical phenomena} in gravity (see~\cite{Gundlach2007} for more details
about the relation between self-similar solutions and critical
phenomena).

As was discussed in the Introduction, the existence of initial data
equations for homothetic Killing vector fields is a natural question.
It was first treated and essentially solved in the early
work~\cite{berger}, quickly following the seminal work~\cite{berezdivin,
MONCRIEF-KID1} work on Killing initial data. Though they did not exactly
use the same language as we did in Section~\ref{sec:propeq}, the early
references~\cite{berezdivin, MONCRIEF-KID1, berger} did derive
propagation identities, but did not provide sufficient conditions to show
that the corresponding equations have a well-posed initial value
problem. They did write the propagation equations in Cauchy-Kovalevskaya
form (solved for the highest order time derivatives), but that is
sufficient for well-posedness only for analytic initial data. In this
section, we reobtain the \emph{homothetic Killing initial data} (HKID)
equations from~\cite[Prop.2]{berger}. However, the way that we derive
propagation allows us to write them directly using wave operators,
making their well-posedness manifest (since they belong to the
well-studied normally-hyperbolic class mentioned in
Section~\ref{sec:propeq}). For Killing initial data this was first done
in~\cite{COLL77}, but to date does not appear to have been done
explicitly for the homothetic case.

The most straightforward way to obtain the relevant propagation
identity is to start with the analogous identity~\eqref{eq:k-propeq} for
the Killing equation and first take its trace,
\begin{equation} \label{eq:trk-propeq}
	\square \K_{c}{}^{c}[v] = \K_{c}{}^{c}[\square v + R\cdot v]
		- 2v^c\nabla_c R - 4 R^{cd} \nabla_c v_d ,
\end{equation}
and then the take gradient of that, noting the relation $\frac{1}{2}
\K_{c}{}^{c}[v] = \nabla^c v_c =: \delta v$ and commuting $\nabla_a$
with $\square$,
\begin{multline} \label{eq:dtrk-propeq}
	\square (\nabla_a \delta v)
	= \nabla_a \delta(\square v + R\cdot v) \\
		- \nabla_a (v^c\nabla_c R) - 2 \nabla_a (R^{cd} \nabla_c v_d)
		+ R_a{}^d \nabla_d \delta v .
\end{multline}
Next, using the identity $\K_{ab}[v] = \CK_{ab}[v] + \frac{1}{n} g_{ab}
\K_{c}{}^{c}[v]$ in~\eqref{eq:k-propeq} and simplifying the result with
the help of~\eqref{eq:trk-propeq}, we obtain the following propagation
identity for the $\CK_{ab}[v]$ operator
\begin{multline}\label{eq:hk-propeq}
	\square \CK_{ab}[v]
		- 2 R^c{}_{ab}{}^d \CK_{cd}[v] - 2 R_{(a}{}^c \CK_{b)c}[v] 
	= \CK_{ab}[\square v + R\cdot v] \\
	+ 2{\mathcal L}_v\left(\frac{R}{n} g_{ab}- R_{ab}\right)
	+ \frac{2}{n} \left(g_{ab}R^{cd}
	- R\delta_a{}^{(c}\delta_b{}^{d)} \right) \K_{cd}[v],
\end{multline}
which is manifestly traceless. Together, Equations~\eqref{eq:hk-propeq}
and~\eqref{eq:dtrk-propeq} make up the propagation identity
($P[E[\phi]] = \sigma[Q[\phi]] + \tau[G]$) for the homothetic Killing
vector equations. The corresponding propagation
equations are
\begin{dgroup*}
\begin{dmath}\label{eq:h-k-v-propeq}
	\square v_a + R_a{}^b v_b = 0
		\condition[]{$(Q[\phi] = 0)$} ,
\end{dmath}
\begin{dmath}\label{eq:h-k-v-propeq2}
	\begin{bmatrix}
	\square h_{ab}
		- 2 R^c{}_{ab}{}^d h_{cd} - 2 R_{(a}{}^c h_{b)c} \\
	\square w_a
	\end{bmatrix}
	= 0
		\condition[]{$(P[\psi] = 0)$} ,
\end{dmath}
\end{dgroup*}
where $h_{ab}$ is considered to be symmetric and traceless, and we must
note that 
\begin{equation} \label{eq:hkv-propeq-conv}
	\square v_a + R_a{}^b v_b
	= \nabla^b \CK_{ab}[v] - \frac{n-2}{n} \nabla_a \delta v
	\quad (Q[\phi] = \rho[E[\phi]]) .
\end{equation}

By performing an analysis similar to that of Section~\ref{sec:kid}, that
is, relying on Lemmas~\ref{lem:propeq} and~\ref{lem:propeq-conv} as well
as the propagation identities~\eqref{eq:hk-propeq},
\eqref{eq:dtrk-propeq} and~\eqref{eq:hkv-propeq-conv}, we can compute
the necessary and sufficient conditions which yield a homothetic Killing
initial data set (HKID) on $\Sigma$, namely any set of equations
equivalent to the following:
\begin{dgroup}
\begin{dmath}
	\left.\CK_{ab}[v]\right|_\Sigma
	\to \left.\begin{bmatrix}
			\CK_{00}[v] & \CK_{0B}[v] \\
			\CK_{0A}[v] & \CK_{AB}[v]
		\end{bmatrix}\right|_\Sigma \hiderel{=} 0 ,
\end{dmath}
\begin{dmath}
	\left.\nabla_0\CK_{ab}[v]\right|_\Sigma
	\to \left.\begin{bmatrix}
			\nabla_0\CK_{00}[v] & \nabla_0\CK_{0B}[v] \\
			\nabla_0\CK_{0A}[v] & \nabla_0\CK_{AB}[v]
		\end{bmatrix}\right|_\Sigma \hiderel{=} 0 ,
\end{dmath}
\begin{dmath}
	\left.\nabla_a \delta v\right|_\Sigma
	\to \left.\begin{bmatrix}
			\nabla_0 \delta v \\ D_A \delta v
		\end{bmatrix}\right|_\Sigma
	\hiderel{=} 0 ,
\end{dmath}
\begin{dmath}\label{eq:lasthkid}
	\text{and} \quad
	\left.\nabla_0 \nabla_a \delta v\right|_\Sigma
	\to \left.\begin{bmatrix}
			\nabla_0 \nabla_0 \delta v \\ \nabla_0 D_A \delta v
		\end{bmatrix}\right|_\Sigma
	\hiderel{=} 0 .
\end{dmath}
\end{dgroup}

For reference, some of the explicit components of the above operators are
\begin{dgroup}
\begin{dmath}\label{eq:hk00}
	\CK_{00}[v] = \frac{2}{n}
	\left[
	(n-1) \nabla_{0}v_{0}
	- \pi v_{0} + D_{A}v^{A}
	\right] ,
\end{dmath}
\begin{dmath}\label{eq:hk0B}
	\CK_{0B}[v] = \nabla_{0}v_{B}
		+ D_{B}v_{0} - \pi_{B}{}^{A} v_{A} ,
\end{dmath}
\begin{dmath}\label{eq:hkAB}
	\CK_{AB}[v] = D_{A}v_{B} + D_{B}v_{A} -2 \pi_{AB} v_{0} \\
		- \frac{2}{n}g_{AB}\left(-\nabla_{0}v_{0} - \pi v_{0} +
		D_{C}v^{C}\right),
\end{dmath}
\begin{dmath} \label{eq:ddv0}
	\nabla_0 \delta v = -\nabla_0 \nabla_0 v_0 - \nabla_0(\pi v_0) \\
		+ D^B \nabla_0 v_B - \pi^{BC} D_B v_C - R_{0 B} v^B ,
\end{dmath}
\begin{dmath} \label{eq:ddvA}
	D_A \delta v = D_A(-\nabla_0 v_0 - \pi v_0 + D_B v^B) ,
\end{dmath}
\end{dgroup}
where we recall that $R_{0B} = D^{C}\pi_{CB} - D_{B}\pi$.

\medskip
\noindent
We are now ready to formulate the main result of this section:
\begin{theorem}\label{thm:hkid}
Consider an $n$-dimensional globally hyperbolic Einstein vacuum
Lorentzian manifold, $(M,g)$ with $R_{ab}=0$, and a Cauchy surface
$\Sigma \subset M$. For $n>2$, the necessary and sufficient conditions
yielding a set of \emph{homothetic Killing initial data} (HKID) for
$v_a$ on $\Sigma$ are given by the following equations:
 
\begin{dgroup} \label{eq:hkid}
\begin{dmath} \label{eq:hkid0}
	D_{A}v_{B} + D_{B}v_{A}
	-2 \pi_{AB} v_{0}
	- \frac{2g_{AB}}{n-1}(D_{C}v^{C} - \pi v_{0}) 
	= 0 \;,
\end{dmath}
\begin{dmath}\label{eq:hkid1}
	D_{A}D_{B}v_{0}
	+ \left(2 (\pi\cdot\pi)_{AB} - \pi \pi_{AB} - r_{AB}\right) v_{0}
	- 2\pi_{C(A} D_{B)}v^{C}
	- v^{C} D_{C}\pi_{AB}
	+ \frac{\pi_{AB}}{n-1}(D_{C}v^{C}-v_0\pi)
	= 0 \;,
\end{dmath}
\begin{dmath}\label{eq:hkid2}
	D_B(D_Av^A-\pi v_0)=0 \;.
\end{dmath}
\end{dgroup}
\end{theorem}

\begin{proof}
The proof is straightforward by direct calculation. The derivatives
$\nabla_0 v_0$, $\nabla_0 v_B$, $\nabla_0\nabla_0 v_0$ and
$\nabla_0\nabla_0 v_B$ are eliminated respectively by $\CK_{00}[v]$,
$\CK_{0B}[v]$, $\nabla_0 \CK_{00}[v]$ and $\nabla_0 \CK_{0B}[v]$. Thus,
the components $\CK_{AB}[v]$ and $D_B \delta v$ respectively lead to the
desired initial data equations~\eqref{eq:hkid0} and~\eqref{eq:hkid2}.
Note that equation~\eqref{eq:hkid0} is manifestly traceless. Thus, we
expect the vanishing $\nabla_0 \CK_{AB}[v]$ to contribute another
traceless equation. It turns out to be convenient to add to it a trace
component proportional to $\nabla_0 \delta v$, thus leading to the last
independent initial data equation~\eqref{eq:hkid1}. The remaining
initial data equations are not independent because of the identities
\begin{dgroup}
\begin{dmath}
	\nabla_0 D_A \delta v
	= D_A (\nabla_0 \delta v) - \pi_A{}^B (D_B \delta v) ,
\end{dmath}
\begin{dmath}
	\nabla_0 \nabla_0 \delta v
	= D^A (D_A \delta v) - \pi (\nabla_0 \delta v)
		+ \nabla_0 (\square v)_0 - D^A (\square v)_A + \pi (\square v)_0 ,
\end{dmath}
\end{dgroup}
where the last identity follows from the splitting
of~\eqref{eq:trk-propeq}, with $(\square v)_0$ and $(\square v)_A$
themselves expressible as
\begin{dgroup}
\begin{dmath}
	(\square v)_0
	= -\nabla_0 \CK_{00}[v] + D^B \CK_{0B}[v] - \pi \CK_{00}[v]
		- \frac{n-2}{n} (\nabla_0 \delta v) ,
\end{dmath}
\begin{dmath}
	(\square v)_A
	= -\nabla_0 \CK_{0A}[v] + D^B \CK_{AB}[v] - \pi \CK_{0A}[v]
		- \pi_{A}{}^{B} \CK_{0B}[v]
		- \frac{n-2}{n} (\nabla_A \delta v) ,
\end{dmath}
\end{dgroup}
due to the splitting of~\eqref{eq:hkv-propeq-conv}.
\end{proof}

Note that a homothetic Killing vector $v_a$ is also a normal Killing
vector exactly when it is divergence free, $\delta v = 0$.
Eliminating $\nabla_0$ derivatives, as in the proof of the theorem,
the divergence can be written as
\begin{dmath} \label{eq:dv}
	\delta v = \nabla_a v^a \hiderel{=} -\nabla_0 v_0 - \pi v_0 + D_A v^A
		= \frac{n}{n-1} \left(D_A v^A - \pi v_0\right)
			+ \frac{n}{2(n-1)} \CK_{00}[v] .
\end{dmath}
Note that we have written the HKID equations~\eqref{eq:hkid} in such a
way that when the spatial divergence free condition
\begin{equation}
	D_A v^A - \pi v_0 = 0
\end{equation}
is satisfied, the HKID equations manifestly reduce to the KID
equations~\eqref{eq:kid}.

\section{Conformal Killing initial data} \label{sec:ckid}

A \emph{conformal Killing} vector $v_a$ satisfies the equation
$\CK_{ab}[v]=0$, defined in~\eqref{eq:ck}. It is less restrictive than
either the Killing or the homothetic Killing equations, discussed in
Sections~\ref{sec:kid} and~\ref{sec:hkid}. As we have seen,
heuristically, the less restrictive the equation, the more complicated
the corresponding propagation identity and the corresponding initial
data equations (if they exist). That pattern will repeat in this
section, where we for the first time both prove the existence of a
propagation identity~\eqref{eq:ck-propeq4} and explicitly construct the
\emph{conformal Killing initial data} (CKID) equations
(Theorem~\ref{thm:ckid}). The structure of this section is modeled on
and uses notation from Sections~\ref{sec:kid} and~\ref{sec:hkid}.

The first attempt to construct the CKID equations (though without using
that terminology) was in the early paper~\cite{berger}, which quickly
followed the original work on the KID equations~\cite{berezdivin,
MONCRIEF-KID1}. However, the construction was not complete and only
obtained some necessary conditions on the initial data of
$v_a$~\cite[Sec.V]{berger}, but without deriving a set of initial data
conditions that could be sufficient. Unfortunately, the CKID problem
does not appear to have been seriously revisited since then. In
Theorem~\ref{thm:ckid}, we finally and for the first time give a complete
construction of the CKID equations (which are both necessary and
sufficient) on Einstein vacuum spacetimes. In retrospect, we can also
answer the following question: why was a full set of initial data
conditions not discovered already in~\cite{berger}? The answer is
simple. The strategy in~\cite{berger} was to split the components
of the conformal Killing equation into evolution equations and spatial
constraint equations, and then take time derivatives of the latter
generating further spatial constraints (modulo the evolution equations),
until hopefully after a certain number of time derivatives no new
spatial constraint equations would be generated, giving an analog of
what we called a propagation identity. The existence of our fourth-order
propagation identity~\eqref{eq:ck-propeq4} implies that this strategy
would have succeeded after four time differentiations (cf.~the
discussion of Equations~\eqref{eq:ckaf-propeq-id} in the proof of
Theorem~\ref{thm:ckid}). Unfortunately, the calculations
in~\cite{berger} stopped at the third time derivative, just short of the
necessary differential order.

There are a couple of points at which the discussion below must deviate
from the parallel Killing case of~\ref{sec:kid}. In particular, to
simplify the various identities to appear below, we make the blanket
assumption that for the rest of this section we are dealing with an
Einstein vacuum background, satisfying $R_{ab}=0$.

The first problem is that the candidate propagation operator
\begin{equation} \label{eq:ck-propeq-conv}
	Q_a[v] := \square v_a + \frac{n-2}{n} \nabla_a \nabla^b v_b
	= \nabla^b \CK_{ab}[v]
	\quad (Q[\phi] = \rho[E[\phi]])
\end{equation}
is no longer normally-hyperbolic, as was the case for the Killing
equation, because in addition to the wave term $\square v_a$ also the
$\nabla_a \nabla^b v_b$ term contributes to the principal symbol.
Fortunately, $Q$ does belongs to the generalized normally-hyperbolic
class that we discussed at the end of Section~\ref{sec:propeq}, and
therefore $Q[\phi]=0$ is a propagation equation.

\begin{lemma} \label{lem:gennormhyp-vec}
Any operator of the form
\begin{equation}
	Q_{a}[v] = x\square v_{a} + (z-x) \nabla_a\nabla^b v_b ,
\end{equation}
where $x\ne 0$ and $z\ne 0$ is generalized normally-hyperbolic.
\end{lemma}
\begin{proof}
The identity
\begin{equation}\label{eq:q-identity}
	\frac{1}{x}\square Q_a[v]
	+ \left(\frac{1}{z}-\frac{1}{x}\right) \nabla_a \nabla^b Q_b[v]
		= \square^2 v_a + \text{l.o.t}
\end{equation}
is all that is needed, which holds precisely when $x\ne 0$ and $z\ne 0$.
\end{proof}

Clearly, our $Q$ from~\eqref{eq:ck-propeq-conv} is a special case of the
operator in Lemma~\ref{lem:gennormhyp-vec} with $x=1 \ne 0$,
$z=\frac{2(n-1)}{n} \ne 0$. Indeed, for these values of $x$ and $z$ a 
computation shows that the identity \eqref{eq:q-identity} adopts 
the form
\begin{equation}\label{eq:q2-identity}
	\square Q_{a}[v] - \frac{(n-2)}{2(n-1)}\nabla_{a}\nabla^{b}Q_{b}[v]
	= 
	\square^2 v_{a}.
\end{equation}

The second problem is that the actual propagation identity for the $\CK$
operator is more convenient to express by coupling it to one of its
integrability conditions, which is propagated separately. Namely, due to
the identity (with the notation $(\delta\CK[v])_a = \nabla^b
\CK_{ab}[v]$)
\begin{dmath} \label{eq:ck2af}
	\nabla_a\nabla_b (\nabla^c v_c)
	= S_{ab}[\CK[v]]
	\coloneq -\frac{n}{2(n-2)} \square \CK_{ab}[v]
		+ \frac{n}{2(n-2)} \CK_{ab}[\delta \CK[v]] \\
		+ \frac{1}{2(n-1)} g_{ab} \nabla^c\nabla^d \CK_{cd}[v]
		- \frac{n}{(n-2)} R_a{}^c{}_b{}^d \CK_{cd}[v] ,
\end{dmath}
when $\CK_{ab}[v]=0$, the divergence $u = \nabla^c v_c$ satisfies the
\emph{covariant affine} equation
\begin{equation}
	\Af_{ab}[u] := \nabla_a\nabla_b u = 0 .
\end{equation}
It satisfies (modulo $R_{ab}=0$) the propagation identities

\begin{equation} \label{eq:af-propeq}
	\square u = \Af^c{}_{c}[u] , \quad
	\square \Af_{ab}[u] + 2 R_a{}^c{}_b{}^d \Af_{cd}[u]
		= \nabla_a \nabla_b \square u.
\end{equation}
Importantly, the propagation equation for $u$ follows from the
propagation equation for $v_a$ because
\begin{equation}
	\square (\delta v) = \frac{n}{2(n-1)} \nabla^c Q_c[v] .
\end{equation}

The coupled propagation identity for the $\CK$ operator then takes the
form
\begin{dmath} \label{eq:ck-propeq2}
	\square \CK_{ab}[v]
	+ 2R_{a}{}^c{}_{b}{}^d \CK_{cd}[v] \\
	+ \frac{2(n-2)}{n}\left(\Af_{ab}[\delta v]
		- \frac{1}{n} g_{ab} \Af^c{}_c[\delta v]\right)
	=\CK_{ab}[Q[v]] .
\end{dmath}

The coupled propagation system~\eqref{eq:ck-propeq2}
and~\eqref{eq:af-propeq} can be combined into a single propagation
identity for the $\CK$ operator, at the expense of making it higher
order:
\begin{dmath} \label{eq:ck-propeq4}
	\square^2 \CK_{ab}[v] + 2 \square (R_a{}^p{}_b{}^q \CK_{pq}[v])
		- \frac{4(n-2)}{n} R_a{}^c{}_b{}^d S_{cd}[\CK[v]] \\
	= \square \CK_{ab}[Q[v]]
		- \frac{(n-2)}{(n-1)} \nabla_a\nabla_b \nabla^c Q_c[v]
		+ \frac{(n-2)}{n(n-1)} g_{ab} \square \nabla^c Q_c[v]
	\\
	{}\hspace{16em} (P[E[\phi]] \hiderel{=} \sigma[Q[\phi]]) .
\end{dmath}
Introducing the trace-free operator $\overline{\Af}_{ab}[u] =
\Af_{ab}[u] - \frac{g_{ab}}{n} \Af_c{}^c[u]$, expanding the definition
of $S_{ab}$ from~\eqref{eq:ck2af} and using basic simplifications, the
propagation identity becomes
\begin{dmath}
	\square^2 \CK_{ab}[v] 
	+ 2\square (R_a{}^c{}_b{}^d \CK_{cd}[v])
	+ 2R_a{}^c{}_b{}^d\square \CK_{cd}[v]
	+ 4  R_a{}^c{}_b{}^d R_c{}^e{}_d{}^f \CK_{ef}[v]
	=
	\square \CK_{ab}[Q[v]] 
	- \frac{(n-2)}{(n-1)}\overline{\Af}_{ab}[\delta Q[v]] 
	\;.
\end{dmath}

The advantage of the higher order propagation
identity~\eqref{eq:ck-propeq4} is that it, together
with~\eqref{eq:ck-propeq-conv}, fits directly into the hypotheses of our
Lemmas~\ref{lem:propeq} and~\ref{lem:propeq-conv}, implying that there
exists a system of \emph{conformal Killing initial data} (CKID)
conditions $\CK^\Sigma[{v|_\Sigma},\nabla_0 {v|_\Sigma}] = 0$, whose
solutions on a Cauchy surface $\Sigma \subset M$ are are in bijection
with solutions of $\CK[v]=0$ on $M$. It remains only to compute it.

It is more practical to carry out this calculation starting from the
coupled second order system of propagation
identities~\eqref{eq:ck-propeq2} and~\eqref{eq:af-propeq}. For that
purpose, let us fix a Cauchy surface $\Sigma \subset M$ and follow the
notational conventions introduced in Section~\ref{sec:kid}. The
components of the conformal Killing and covariant affine operators
\begin{dgroup}
\begin{dmath}
	\left.\CK_{ab}[v]\right|_\Sigma
	\to \left.\begin{bmatrix}
			\CK_{00}[v] & \CK_{0B}[v] \\
			\CK_{0A}[v] & \CK_{AB}[v]
		\end{bmatrix}\right|_\Sigma \hiderel{=} 0
\end{dmath}
\begin{dmath}
	\text{and} \quad
	\left.\Af_{ab}[u]\right|_\Sigma
	\to \left.\begin{bmatrix}
			\Af_{00}[u] & \Af_{0B}[u] \\
			\Af_{0A}[u] & \Af_{AB}[u]
		\end{bmatrix}\right|_\Sigma \hiderel{=} 0
\end{dmath}
\end{dgroup}
take the explicit form
\begin{dgroup}
\begin{dmath}\label{eq:ck00}
	\CK_{00}[v] = \frac{2}{n}
	\left[
	(n-1) \nabla_{0}v_{0}
	- \pi v_{0} + D_{A}v^{A}
	\right] ,
\end{dmath}
\begin{dmath}\label{eq:ck0B}
	\CK_{0B}[v] = \nabla_{0}v_{B}
		+ D_{B}v_{0} - \pi_{B}{}^{A} v_{A} ,
\end{dmath}
\begin{dmath}\label{eq:ckAB}
	\CK_{AB}[v] = D_{A}v_{B} + D_{B}v_{A} -2 \pi_{AB} v_{0} \\
		- \frac{2}{n}g_{AB}\left(D_{C}v^{C} - \pi v_{0}
			- \nabla_{0}v_{0} \right),
\end{dmath}
\begin{dmath} \label{eq:af00}
	\Af_{00}[u] = \nabla_0\nabla_0 u ,
\end{dmath}
\begin{dmath} \label{eq:af0B}
	\Af_{0B}[u] = D_B(\nabla_0 u) - \pi_{BC} D^C u ,
\end{dmath}
\begin{dmath} \label{eq:afAB}
	\Af_{AB}[u] = D_A D_B u - \pi_{AB} (\nabla_0 u) .
\end{dmath}
\end{dgroup}

As a first step, completely analogous to the review of the Killing
vector case in Section~\ref{sec:kid}, we have the following
\begin{theorem} \label{thm:afid}
Consider a globally hyperbolic Einstein vacuum Lorentzian manifold,
$(M,g)$ of dimension $n>0$ with $R_{ab}=0$, and a Cauchy surface $\Sigma
\subset M$. The necessary and sufficient conditions yielding a set of
\emph{covariant affine initial data} (AfID) for $u$ on $\Sigma$ are
given by the following equations:
\begin{dgroup} \label{eq:afid}
\begin{dmath} \label{eq:afid0B}
	D_B(\nabla_0 u) - \pi_{BC} D^C u = 0 \;,
\end{dmath}
\begin{dmath} \label{eq:afid0AB}
	D_A D_B u - \pi_{AB} (\nabla_0 u) = 0 \;,
\end{dmath}
\begin{dmath}\label{eq:afid1AB}
	(r_{A B} + \pi\pi_{A B} -(\pi\cdot\pi)_{AB}) (\nabla_0 {u})
	\\
	-(D_{(A}{\pi_{B)C}} - D_{C}{\pi_{A B}}) D^{C}{u}
	= 0 \;.
\end{dmath}
\end{dgroup}
\end{theorem}
\begin{proof}
The components $\Af_{0B}[u]$ and $\Af_{AB}[u]$ directly
give~\eqref{eq:afid0B} and~\eqref{eq:afid0AB}. The derivative $\nabla_0
\Af_{AB}[u]$ gives~\eqref{eq:afid1AB}, after using
$\Af_{00}[u]$ to eliminate $\nabla_0\nabla_0 u$, and further
simplifications from the first two initial data conditions. Lastly, the
condition $\nabla_0 \Af_{0B}[u]=0$ turns out not to be independent from
the other ones due to the identity
\begin{equation}
	\nabla_0 \Af_{0B}[u]
	= 2 g^{CA} D_{[C} \left(D_{B]} D_A u - \pi_{B]A} (\nabla_0 u)\right) ,
\end{equation}
after simplifications from $\Af_{00}[u]=0$ and the vacuum Einstein
equations.
\end{proof}

So, now what we need to do is start with the conditions
${\CK[v]|_\Sigma} = 0$, \ldots, ${\nabla_0^3\CK[v]|_\Sigma} = 0$, and
extract from them a (hopefully small) subset of components whose
vanishing ensures the vanishing of the remaining components as well.

\begin{theorem} \label{thm:ckid}
Consider a globally hyperbolic Einstein vacuum Lorentzian manifold,
$(M,g)$ of dimension $n>2$ with $R_{ab}=0$, and a Cauchy surface $\Sigma
\subset M$. For a conformal Killing vector $v_a$, $\CK_{ab}[v]=0$, its
rescaled divergence $u = \frac{(n-1)}{n}\nabla^a v_a$ and its derivative
$\nabla_0 u$ take the following form when restricted to $\Sigma$, after
eliminating the $\nabla_0 v_a$ derivatives,
\begin{dgroup} \label{eq:div-v}
\begin{dmath}
	u = \left(D_C v^C - \pi v_0\right) \;,
\end{dmath}
\begin{dmath}
	\nabla_0 u
		= \frac{1}{n-1} \pi u + \left(-D^{A}{D_{A}{v_{0}}}
			+ (\pi\cdot\pi) v_{0} +(D^{A}{\pi}) v_{A}\right) \;.
\end{dmath}
\end{dgroup}
Using the above notation, the necessary and sufficient conditions
yielding a set of \emph{conformal Killing initial data} (CKID) for $v_a$
on $\Sigma$ are given by the following equations:
\begin{dgroup} \label{eq:ckid}
\begin{dmath} \label{eq:ckid0}
	D_{A}v_{B} + D_{B}v_{A} -2 \pi_{AB} v_{0}
		- \frac{2}{n-1}g_{AB}
				u 
	= 0 \;,
\end{dmath}
\begin{dmath} \label{eq:ckid1}
	D_{B}D_{A}v_{0}
	+ (2 (\pi\cdot\pi)_{AB} - \pi\pi_{AB} - r_{AB}) v_{0}
	- 2\pi_{C(A} D_{B)}v^{C} + v^{C}(D_{C}\pi_{AB})
	+ \frac{1}{n-1}(u\pi_{AB}+g_{AB}\nabla_{0}u)
	= 0 \;,
\end{dmath}

\begin{dmath} \label{eq:ckidaf0AB}
	D_A D_B u 
		- \pi_{AB} (\nabla_0 u) 
	= 0 ,
\end{dmath}
\begin{dmath} \label{eq:ckidaf1AB}
	(r_{A B} + \pi\pi_{A B} -(\pi\cdot\pi)_{AB}) (\nabla_0 u) 
	\\
		-(D_{(A}{\pi_{B)C}} - D_{C}{\pi_{A B}}) D^{C} u 
	= 0 \;.
\end{dmath}
\end{dgroup}
\end{theorem}
\begin{proof}
Schematically (modulo the $Q_a[v]=0$ propagation equation), the
propagation identities~\eqref{eq:ck-propeq2}
and~\eqref{eq:af-propeq} imply, respectively, the initial data
identities
\begin{dgroup} \label{eq:ckaf-propeq-id}
\begin{dmath} \label{eq:ck-propeq2-id}
	\nabla_0\nabla_0 \CK = O(\Af) + O(\nabla_0 \CK) + O(\CK) ,
\end{dmath}
\begin{dmath} \label{eq:af-propeq-id}
	\nabla_0\nabla_0 \Af = O(\nabla_0 \Af) + O(\Af) ,
\end{dmath}
\intertext{%
where $\CK$ and $\Af$ stand respectively for the components of $\CK[v]$
and $\Af[\delta v]$, while the notation $O(-)$ indicates proportionality
to the argument or any spatial derivative thereof. Taking $\nabla_0$
derivatives of~\eqref{eq:ck-propeq2-id} and
using~\eqref{eq:af-propeq-id} to eliminating as many higher order
$\nabla_0$ derivatives of $\CK$ and $\Af$ as possible, we obtain the
further initial data identities%
}
\begin{dmath}
	\nabla_0^3 \CK = O(\nabla_0 \Af) + O(\Af) + O(\nabla_0 \CK) + O(\CK) ,
\end{dmath}
\begin{dmath}
	\nabla_0^4 \CK = O(\nabla_0 \Af) + O(\Af) + O(\nabla_0 \CK) + O(\CK) .
\end{dmath}
\end{dgroup}
Thus, it is sufficient to keep only the $\CK$, $\nabla_0 \CK$, $\Af$,
$\nabla_0 \Af$ components for the CKID, and by Lemmas~\ref{lem:propeq}
and~\ref{lem:propeq-conv} the solutions of the resulting CKID conditions
on $\Sigma \subset M$ would be in bijection with the solutions of the
conformal Killing equation on $M$.

The $\CK_{00}[v]$ and $\CK_{0B}[v]$ components can be used to eliminate
any $\nabla_0$ derivatives of $v_0$ and $v_B$, respectively, and thus
they and their $\nabla_0$ derivatives need not appear in the final CKID
system. The $\CK_{AB}[v]$ and $\nabla_0\CK_{AB}[v]$ components, after
eliminating the $\nabla_0$ derivatives, give
respectively~\eqref{eq:ckid0} and~\eqref{eq:ckid1}, where
also~\eqref{eq:ckid0} was used to simplify the form of~\eqref{eq:ckid1}.

Theorem~\ref{thm:afid} has already shown that the vanishing of $\Af$ and
$\nabla_0\Af$ are equivalent to the vanishing of $\Af_{00}$ and the AfID
system~\eqref{eq:afid}. It remains only to plug in the following
expressions, with
\begin{equation}
 u = \frac{(n-1)}{n} \delta v
	 = \frac{(n-1)}{n} \left(-\nabla_0 v_0 - \pi v_0 + D_C v^C\right) ,
\end{equation}
which by direct calculation, after eliminating the
$\nabla_0$ derivatives of $v_0$ and $v_B$ using $\CK_{00}[v]$ and
$\CK_{0B}[v]$, leads to the expressions in~\eqref{eq:div-v}.
The resulting $\Af_{00}[\delta v]$
and~$\Af_{0B}[\delta v]$ expressions are not independent, due to the
identities
\begin{dmath}
	\Af_{00}[\delta v] = \frac{n}{2(n-1)(n-2)} \left[
		-(2n-3) \pi^{AC}\nabla_0 \CK_{AC}[v] \\
		+ D^A D^C \CK_{AC}[v] -r^{AC} \CK_{AC}[v]
	\right] \;,
\end{dmath}
\begin{dmath}
	\Af_{0B}[\delta v] = \frac{n}{2(n-2)} \left[
		D^A\nabla_0 \CK_{AB}[v] + \pi^{AC} D_A\CK_{CB}[v] - \pi_B{}^C D^A \CK_{AC}[v]
		- \pi^{AC} D_B \CK_{AC}[v]
		+ (D^C \pi) \CK_{CB}[v] - (D^A \pi^C{}_B) \CK_{AC}[v]
	\right] \;,
\end{dmath}
again modulo $\CK_{00}[v]=0$ and $\CK_{0B}[v]=0$, which are obtained
by splitting the spacetime identity~\eqref{eq:ck2af}. The
conditions~\eqref{eq:afid0AB} and~\eqref{eq:afid1AB}, after the $u$ and
$\nabla_0 u$ substitution, directly give respectively the remaining CKID
conditions~\eqref{eq:ckidaf0AB} and~\eqref{eq:ckidaf1AB}, which
completes the proof.
\end{proof}

Obviously, a Killing vector $v$, is a conformal Killing vector
satisfying the extra divergence condition $\delta v = 0$. As we have
seen in the above proof, according to~\eqref{eq:div-v}, the vanishing of
the divergence $\delta v$ and its derivative $\nabla_0 \delta v$ are
equivalent to the initial data conditions
\begin{dgroup}
\begin{dmath}
	D_C v^C - \pi v_0 = 0 \;,
\end{dmath}
\begin{dmath}
	-D^{A}{D_{A}{v_{0}}}
			+ (\pi\cdot\pi) v_{0} +(D^{A}{\pi}) v_{A} = 0 \;,
\end{dmath}
\end{dgroup}
once $\nabla_0$ derivatives have been eliminated using $\CK_{00}[v]=0$
and $\CK_{0B}[v]=0$. Thus, when the initial data for $v_a$ is divergence
free in the above sense, it is obvious that the CKID
conditions~\eqref{eq:ckidaf0AB} and~\eqref{eq:ckidaf1AB} are
tautological, while the conditions~\eqref{eq:ckid0} and~\eqref{eq:ckid0}
recover the KID system~\eqref{eq:kid}, as was to be expected.

\section{Discussion}
We have presented for the first time in the literature a set of
necessary and sufficient conditions (the CKID equations) ensuring that a
vacuum initial data set of the Einstein's equations in any dimension
($n>2$) admits a conformal Killing vector in any globally hyperbolic
development of this initial data. In addition to the standard quantities
required for the construction of vacuum initial data (the first and the
second fundamental forms, given respectively by $g_{AB}$, $\pi_{AB}$ in
our notation) we need the {\em conformal Killing lapse} $v_0$ and {\em
conformal Killing shift} $v_A$. The CKID conditions are given by
\eqref{eq:ckid} of Theorem~\ref{thm:ckid} and they are a set of linear
PDEs for $v_0, v_A$ on the Riemannian manifold with extrinsic curvature
($\Sigma$, $g_{AB}$, $\pi_{AB}$). Along the way, we have reviewed
construction of the Killing initial data (KID) and gave a new derivation
of the homothetic Killing initial data (HKID) equations. Just as in the
KID case, the HKID and CKID equations likely constitute an
overdetermined elliptic system for $v_0, v_A$, but the true extent of
this assertion requires a separate investigation.

A natural continuation of this work would be to try to construct initial
data systems for other geometric PDEs, like for instance Killing-Yano
equations, higher rank Killing tensor equations, and their conformal
and/or closed versions. For instance, the existence of a principal
(closed and non-degenerate) conformal Killing-Yano 2-form is known to
characterize the Kerr-NUT-(A)dS family of higher dimensional black holes
and related solutions~\cite{fkk-review}. So it is reasonable to suppose
that the knowledge of the corresponding initial data system could be of
use in the study of the stability and rigidity of this family. In $4$
spacetime dimensions, the conformal Killing-Yano 2-form equation is
equivalent to the Killing $(2,0)$-spinor equation~\cite{mclvdb}, whose
initial data system was already constructed in~\cite{GOMEZLOBO2008}. A
tensorial version of this initial data system will appear in future
work, along with an extension to the closed conformal Killing-Yano
2-form case in higher dimensions. The question of which other variations
of the Killing equations have initial data systems appears to be
completely open.

Since, in $4$ spacetime dimensions, the conformal Killing equation is
equivalent to the Killing $(1,1)$-spinor
equation~\cite{vanNieuwenhuizen1984}, it would be interesting to
translate our CKID system into the initial data conditions for Killing
$(1,1)$-spinors. Alternatively, such initial data conditions in $4$
dimensions could be rederived from the relevant spinorial propagation
identity used as an intermediate result in~\cite{GOMEZLOBO2008}.

In all the known cases where initial data systems have been found, a
certain amount of trial and error has been necessary for success. It
would be an interesting problem to find a systematic way to identify
those cases where no initial data system can exist.

\paragraph{Acknowledgements}
The authors thank Piotr Chru\'sciel, Josef \v{S}ilhan and Juan A.\ 
Valiente Kroon for interesting comments, as well as Marc Mars for
pointing out that a previous version of our main result was incomplete.
IK was partially supported by the Praemium Academiae
of M.~Markl, GA\v{C}R project 19-09659S and RVO: 67985840.
AGP is supported by  GA\v{C}R project 19-01850S.
AGP also thanks the partial support from the projects
IT956-16 (``Eusko Jaur\-la\-ri\-tza'', Spain) and
PTDC/MAT-ANA /1275/2014\\ 
(``Funda\c{c}\~{a}o para a Ci\^{e}ncia e a Tecnologia'', Portugal).

\bibliographystyle{utphys-alpha}
\bibliography{kid}

\end{document}